\documentclass[conference,a4paper]{IEEEtran}

\usepackage{graphicx,color}   
\usepackage[cmex10]{amsmath}  
\usepackage{amssymb}
\usepackage[utf8]{inputenc}
\usepackage[T1]{fontenc}

\newtheorem{theorem}{\bfseries Theorem}
\newtheorem{remark}{Remark}
\newtheorem{lemma}{Lemma}

\long\def\symbolfootnote[#1]#2{\begingroup%
\def\thefootnote{\fnsymbol{footnote}}\footnote[#1]{#2}\endgroup} 

\def\compress{\vspace{-0.7ex}}

\renewcommand{\thefootnote}{\fnsymbol{footnote}}

\begin{document}

\title{The Capacity of a Class of Multi-Way Relay Channels}
\author{\authorblockN{Lawrence Ong, Sarah J. Johnson, and Christopher M. Kellett}
\authorblockA{School of Electrical Engineering and Computer Science,
The University of Newcastle\\
Email: lawrence.ong@cantab.net; \{sarah.johnson, chris.kellett\}@newcastle.edu.au}
}

\maketitle

\begin{abstract}
The capacity of a class of multi-way relay channels, where $L$ users communicate via a relay (at possibly different rates), is derived for the case where the channel outputs are modular sums of the channel inputs and the receiver noise. The cut-set upper bound to the capacity is shown to be achievable. More specifically, the capacity is achieved using (i) rate splitting, (ii) functional-decode-forward, and (iii) joint source-channel coding. We note that while separate source-channel coding can achieve the common-rate capacity, joint source-channel coding is used to achieve the capacity for the general case where the users are transmitting at different rates.
\end{abstract}

\IEEEpeerreviewmaketitle

\section{Introduction}\symbolfootnote[0]{This work is supported by the Australian Research Council under grants DP0877258 and DP1093114.}
We consider the multi-way relay channel (MWRC), where $L$ users ($L \geq 2$) exchange data via a relay, and where there is no direct link between the users. Common applications of this model are conference calls in the cellular network and satellite communications.

The MWRC is an extension of the two-way relay channel (TWRC) where two users ($L=2$) exchange data via a relay (e.g., see \cite{knopp06,rankovwittneben06,rankovwittneben07}). The Gaussian MWRC, where the channels between the nodes are additive white Gaussian noise channels, was first investigated by Gündüz \emph{et al.}~\cite{gunduzyener09}. An upper bound and a few achievable rate regions, based on the coding strategies for the relay channel, were derived using: (i) \emph{complete-decode-forward} (CDF) where the relay completely decodes the users' messages and broadcasts a function of the messages back to the users, (ii) \emph{compress-forward} where the relay quantizes its received signals, re-encodes and broadcasts them to the users, and (iii) \emph{amplify-forward} where the relay simply scales and forwards what it receives.
These coding strategies, however, fail to achieve the capacity of the MWRC.

Recently, \emph{functional-decode-forward} (FDF) has been proposed for the TWRC, where the relay decodes a function of the users' messages and broadcasts the function back to the users. 
FDF has been shown to achieve within $\frac{1}{2}$ bit of the capacity of the Gaussian TWRC~\cite{namchung08}. We later proposed FDF for the multi-way relay channel (MWRC), and showed that FDF achieves the \emph{common-rate} (where all users exchange information at the same rate) capacity of the binary MWRC~\cite{ongjohnsonkellett10cl}, where the channels are binary symmetric. Applying insights from the binary MWRC has allowed us to obtain the common-rate capacity of the the Gaussian MWRC with three or more users where all nodes transmit at the same power~\cite{ong10amwrc}. The ``general'' capacity (i.e., where users can transmit at possibly different rates) of the MWRC is not yet known.

In this paper, we work toward this goal by deriving the ``general'' capacity of the \emph{finite field adder} MWRC, where the channel outputs are the summation (in finite field arithmetic) of the channel inputs and the receiver noise. We show that the capacity can be achieved by combining the ideas of (i) rate splitting, (ii) our proposed FDF~\cite{ongjohnsonkellett10cl}, and (iii) the joint source-channel coding for broadcast channels by Tuncel \cite{tuncel06}. This, to the best of our knowledge, is the first example of the MWRC where the capacity is found for all noise distributions/levels. 

The rest of the paper is organized as follows. We define the channel model of the finite field adder MWRC in Sec.~\ref{model}, and find a capacity upper bound in Sec.~\ref{ub}. In Sec.~\ref{linear}, we construct a linear code that is optimal for the point-to-point finite field adder channel. Using this linear code, we propose a coding strategy using the ideas of rate splitting, FDF, and joint source-channel coding to obtain the capacity of the finite field adder MWRC in Sec.~\ref{fdf}. Lastly, in Sec.~\ref{compare}, we compare the capacity with two other coding strategies, namely FDF with rate splitting and separate source-channel coding and CDF, and discuss why these two strategies fall short of the capacity.

\section{Channel Model}\label{model}

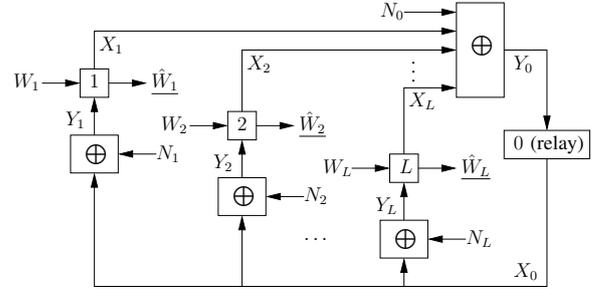
\begin{figure}[t]
\centering
\resizebox{7.7cm}{!}{
\begin{picture}(0,0)%
\includegraphics{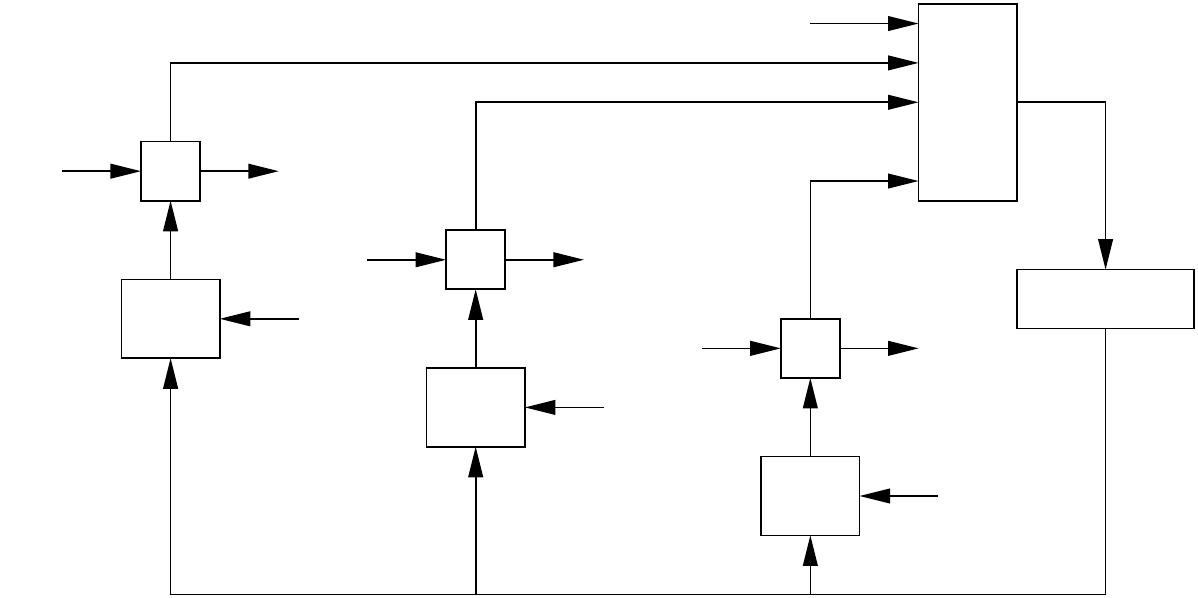}%
\end{picture}%
\setlength{\unitlength}{4144sp}%
\begingroup\makeatletter\ifx\SetFigFont\undefined%
\gdef\SetFigFont#1#2#3#4#5{%
  \fontsize{#1}{#2pt}%
  \fontfamily{#3}\fontseries{#4}\fontshape{#5}%
  \selectfont}%
\fi\endgroup%
\begin{picture}(5472,2724)(301,-2413)
\put(1036,-511){\makebox(0,0)[lb]{\smash{{\SetFigFont{12}{14.4}{\familydefault}{\mddefault}{\updefault}{\color[rgb]{0,0,0}$1$}%
}}}}
\put(316,-511){\makebox(0,0)[lb]{\smash{{\SetFigFont{12}{14.4}{\familydefault}{\mddefault}{\updefault}{\color[rgb]{0,0,0}$W_1$}%
}}}}
\put(1621,-511){\makebox(0,0)[lb]{\smash{{\SetFigFont{12}{14.4}{\familydefault}{\mddefault}{\updefault}{\color[rgb]{0,0,0}$\underline{\hat{W}_1}$}%
}}}}
\put(5041,-1096){\makebox(0,0)[lb]{\smash{{\SetFigFont{12}{14.4}{\familydefault}{\mddefault}{\updefault}{\color[rgb]{0,0,0}$0$  (relay)}%
}}}}
\put(3961,-1321){\makebox(0,0)[lb]{\smash{{\SetFigFont{12}{14.4}{\familydefault}{\mddefault}{\updefault}{\color[rgb]{0,0,0}$L$}%
}}}}
\put(3241,-1321){\makebox(0,0)[lb]{\smash{{\SetFigFont{12}{14.4}{\familydefault}{\mddefault}{\updefault}{\color[rgb]{0,0,0}$W_L$}%
}}}}
\put(4546,-1321){\makebox(0,0)[lb]{\smash{{\SetFigFont{12}{14.4}{\familydefault}{\mddefault}{\updefault}{\color[rgb]{0,0,0}$\underline{\hat{W}_L}$}%
}}}}
\put(2431,-916){\makebox(0,0)[lb]{\smash{{\SetFigFont{12}{14.4}{\familydefault}{\mddefault}{\updefault}{\color[rgb]{0,0,0}$2$}%
}}}}
\put(1711,-916){\makebox(0,0)[lb]{\smash{{\SetFigFont{12}{14.4}{\familydefault}{\mddefault}{\updefault}{\color[rgb]{0,0,0}$W_2$}%
}}}}
\put(3016,-916){\makebox(0,0)[lb]{\smash{{\SetFigFont{12}{14.4}{\familydefault}{\mddefault}{\updefault}{\color[rgb]{0,0,0}$\underline{\hat{W}_2}$}%
}}}}
\put(1126,-151){\makebox(0,0)[lb]{\smash{{\SetFigFont{12}{14.4}{\familydefault}{\mddefault}{\updefault}{\color[rgb]{0,0,0}$X_1$}%
}}}}
\put(4141,-466){\rotatebox{90.0}{\makebox(0,0)[lb]{\smash{{\SetFigFont{12}{14.4}{\familydefault}{\mddefault}{\updefault}{\color[rgb]{0,0,0}$\dotsm$}%
}}}}}
\put(3781,164){\makebox(0,0)[lb]{\smash{{\SetFigFont{12}{14.4}{\familydefault}{\mddefault}{\updefault}{\color[rgb]{0,0,0}$N_0$}%
}}}}
\put(811,-871){\makebox(0,0)[lb]{\smash{{\SetFigFont{12}{14.4}{\familydefault}{\mddefault}{\updefault}{\color[rgb]{0,0,0}$Y_1$}%
}}}}
\put(4051,-691){\makebox(0,0)[lb]{\smash{{\SetFigFont{12}{14.4}{\familydefault}{\mddefault}{\updefault}{\color[rgb]{0,0,0}$X_L$}%
}}}}
\put(3736,-1681){\makebox(0,0)[lb]{\smash{{\SetFigFont{12}{14.4}{\familydefault}{\mddefault}{\updefault}{\color[rgb]{0,0,0}$Y_L$}%
}}}}
\put(4591,-1996){\makebox(0,0)[lb]{\smash{{\SetFigFont{12}{14.4}{\familydefault}{\mddefault}{\updefault}{\color[rgb]{0,0,0}$N_L$}%
}}}}
\put(2206,-1276){\makebox(0,0)[lb]{\smash{{\SetFigFont{12}{14.4}{\familydefault}{\mddefault}{\updefault}{\color[rgb]{0,0,0}$Y_2$}%
}}}}
\put(2521,-331){\makebox(0,0)[lb]{\smash{{\SetFigFont{12}{14.4}{\familydefault}{\mddefault}{\updefault}{\color[rgb]{0,0,0}$X_2$}%
}}}}
\put(1666,-1186){\makebox(0,0)[lb]{\smash{{\SetFigFont{12}{14.4}{\familydefault}{\mddefault}{\updefault}{\color[rgb]{0,0,0}$N_1$}%
}}}}
\put(5041,-331){\makebox(0,0)[lb]{\smash{{\SetFigFont{12}{14.4}{\familydefault}{\mddefault}{\updefault}{\color[rgb]{0,0,0}$Y_0$}%
}}}}
\put(5041,-2311){\makebox(0,0)[lb]{\smash{{\SetFigFont{12}{14.4}{\familydefault}{\mddefault}{\updefault}{\color[rgb]{0,0,0}$X_0$}%
}}}}
\put(3061,-1996){\makebox(0,0)[lb]{\smash{{\SetFigFont{12}{14.4}{\familydefault}{\mddefault}{\updefault}{\color[rgb]{0,0,0}$\dotsm$}%
}}}}
\put(3061,-1591){\makebox(0,0)[lb]{\smash{{\SetFigFont{12}{14.4}{\familydefault}{\mddefault}{\updefault}{\color[rgb]{0,0,0}$N_2$}%
}}}}
\put(991,-1186){\makebox(0,0)[lb]{\smash{{\SetFigFont{12}{14.4}{\familydefault}{\mddefault}{\updefault}{\color[rgb]{0,0,0}$\bigoplus$}%
}}}}
\put(2386,-1591){\makebox(0,0)[lb]{\smash{{\SetFigFont{12}{14.4}{\familydefault}{\mddefault}{\updefault}{\color[rgb]{0,0,0}$\bigoplus$}%
}}}}
\put(4636,-151){\makebox(0,0)[lb]{\smash{{\SetFigFont{12}{14.4}{\familydefault}{\mddefault}{\updefault}{\color[rgb]{0,0,0}$\bigoplus$}%
}}}}
\put(3916,-1996){\makebox(0,0)[lb]{\smash{{\SetFigFont{12}{14.4}{\familydefault}{\mddefault}{\updefault}{\color[rgb]{0,0,0}$\bigoplus$}%
}}}}
\end{picture}%

}
\caption{The finite field adder MWRC}
\label{fig:ff-mwrc}
\end{figure}

Fig.~\ref{fig:ff-mwrc} depicts the MWRC considered in this paper, where there is no direct user-to-user link. Nodes 1, 2, $\dotsc$, $L$ are the users, and node $0$ is the relay. We consider full data exchange where each user is to decode the messages from all other users.  We denote by $X_i$ node $i$'s input to the channel, $Y_i$ the channel output received by node $i$, $W_i$ node $i$'s message, and $\underline{\hat{W}_i}$ node $i$'s estimate of all other users' messages.

The $L$-user finite field adder MWRC over the finite field $\mathcal{F}$ consists of the following:\\
$\bullet$ Uplink: $Y_0 = \left(\bigoplus\limits_{1 \leq i \leq L} X_i \right) \oplus N_0 \triangleq  X_1 \oplus X_2 \oplus \dotsm \oplus X_L \oplus N_0$,\\
$\bullet$ Downlink: $Y_i = X_0 \oplus N_i$, for each $i=1$, $2$, $\dotsc$, $L$,\\
where $X_i, Y_i, N_i \in \mathcal{F}$, $\forall i$, for some finite field $\mathcal{F}$, $\oplus$ is the addition operation associated with $\mathcal{F}$, $N_i$ are statistically independent for all $i$ and all channel uses.
Let $W_i \in \left\{1,2,\dotsc,2^{nR_i}\right\}$ be an $(nR_i)$-bit message, where $R_i$ is a rational number for every $1 \leq i \leq L$, and consider $n$ simultaneous uplink and downlink channel uses. User $i$'s transmit message at time $t$, $X_i[t]$, can only depend on its own message and its past received signals, i.e., $X_i[t] = f_{i,t}(W_i, Y_i[1],Y_i[2],\dotsc,Y_i[t-1])$, for $1 \leq t \leq n$. The relay's transmitted signal at any time can only depend on its past received signals, i.e., $X_0[t] = f_{0,t}(Y_0[1],Y_0[2],\dotsc,Y_0[t-1])$, for $1 \leq t \leq n$. After $n$ channel uses, user $i$ estimates the messages of all other nodes from its received signals and its own message, i.e., $\underline{\hat{W}_i} = g_i(\boldsymbol{Y}_i,W_i)$, where $\boldsymbol{Y}_i = (Y_i[1],Y_i[2],\dotsc,Y_i[n])$. Assume that the users' messages are independent and each $W_i$ is uniformly distributed over $\left\{1,2,\dotsc,2^{nR_i}\right\}$.  We say that the rate tuple $(R_1,R_2,\dotsc,R_L)$ is \emph{achievable} if there exists some $(n,\{f_{i,t}\}_{0 \leq i \leq L, 1 \leq t \leq n}, \{g_i\}_{1 \leq i \leq L})$  such that all users can \emph{reliably} decode the messages of all other users. We say that a user can decode a message reliably if the probability that it wrongly decodes the message can be made arbitrarily small. The \emph{capacity} is defined as the closure of all achievable rate tuples.

\section{A Capacity Upper Bound}\label{ub}

In this section, we derive an upper bound to the capacity of the finite field adder MWRC using cut-set arguments. A cut-set upper bound to the capacity of a network is the maximum rate that information can be transferred across a \emph{cut} separating two disjoint sets of nodes, assuming that all nodes on each side of the cut can fully cooperate.
We define $R_\text{min} = \min\limits_{1 \leq j \leq L} R_j$, $R_i^c = \sum\limits_{j=1,j\neq i}^LR_j$, and $R_\text{min}^c = \left(\sum\limits_{j=1}^L R_j \right) - R_\text{min}$.
The cut-set upper bound to the capacity of the finite field adder MWRC is given in the following theorem.
\begin{theorem}\label{theorem:upper-bound}
Consider an $L$-user finite field adder MWRC over $\mathcal{F}$. If the rate tuple $(R_1,R_2,\dotsc,R_L)$ is achievable, then
\begin{align}
R_\text{min}^c &\leq \log_2|\mathcal{F}| - H(N_0) \label{eq:mwrc-ub-2}\\
R_i^c &\leq \log_2|\mathcal{F}| - H(N_i), \label{eq:mwrc-ub-3}
\end{align}
for all $1 \leq i \leq L$.
\end{theorem}

Here, $H(X)=-\sum_{x \in \mathcal{X}} p(x)\log_2p(x)$ is the entropy.

\begin{proof}[Proof of Theorem~\ref{theorem:upper-bound}]
Consider a network of $m$ nodes, in which node $i$ sends information at the rate $R_{i,j}$ (in bits/channel use) to node $j$. If the set of rates $\{R_{i,j}\}$  are achievable, there exists some joint probability distribution $p(x_1,x_2,\dotsc,x_m)$ such that $\sum_{i \in \mathcal{S}, j \in \mathcal{S}^c} R_{i,j} \leq I(X_\mathcal{S};Y_{\mathcal{S}^c} | X_{\mathcal{S}^c} )$, for all $\mathcal{S} \subset \{1,2,\dotsc,m\}$~\cite[p. 589]{coverthomas06}. Here $X_{\mathcal{S}} = \{X_i: i \in \mathcal{S}\}$, and $\mathcal{S}^c = \{1,2,\dotsc,m\} \setminus \mathcal{S}$.

First, consider the cut separating $\mathcal{S} = \{1,2,\dotsc,i-1,i+2,\dotsc,L\}$ and $\mathcal{S}^c = \{0,i\}$ in the MWRC, for some $1 \leq i \leq L$. An upper bound to the rate $R_i^c$ (of messages $(W_1,W_2,\dotsc,W_{i-1},W_{i+1},\dotsc,W_L)$) across the cut from $\mathcal{S}$ to $\mathcal{S}^c$ is therefore
\begin{subequations}
\begin{align}
&\sum\limits_{j=1,j\neq i}^LR_j = R_i^c \leq I(X_{[1,L]\setminus \{i\}};Y_{\{0,i\}}|X_{\{0,i\}}).\\
&= H(Y_0,Y_i|X_0,X_i ) - H(Y_0,Y_i|X_{[0,L]})\\
&= H(X_1 \oplus \dotsm \oplus X_{i-1} \oplus X_{i+1} \oplus \dotsm \oplus X_L \oplus N_0, N_i) \nonumber\\
&\quad - H(N_0,N_i)\\
&= H\left( \Bigg(\bigoplus_{j \in [1,L]\setminus\{i\}} X_j \Bigg) \oplus N_0 \right) - H(N_0),\label{eq:ub-1}
\end{align}
\end{subequations}
where \eqref{eq:ub-1} is because $\left((\bigoplus_{j \in [1,L]\setminus\{i\}} X_j) \oplus N_0\right)$ and $N_i$ are statistically independent, so are $N_0$ and $N_i$.

Next, consider the cut separating $\mathcal{S} = \{0,1,2,\dotsc,i-1,i+2,\dotsc,L\}$ and $\mathcal{S}^c = \{i\}$, for some $1 \leq i \leq L$. We have the following rate constraint
\begin{subequations}
\begin{align}
R_i^c &\leq I(X_{[0,L]\setminus\{i\}};Y_i|X_i)\\
&= H(X_0 \oplus N_i) - H(N_i). \label{eq:ub-2}
\end{align}
\end{subequations}

The rate constraints \eqref{eq:ub-1} and \eqref{eq:ub-2} must be satisfied for all $1 \leq i \leq L$ for some $p(x_0,x_1,\dotsc,x_L)$. Note that choosing the independent and uniform distribution for each $X_i$, for $0 \leq i \leq L$, simultaneously maximizes all the mutual information terms in the constraints. So, combining the above rate constraints, we have Theorem~\ref{theorem:upper-bound}. Note that \eqref{eq:mwrc-ub-2} implies \eqref{eq:ub-1} for all $1 \leq i \leq L$, since $R_\text{min}^c = \max\limits_{1 \leq j \leq L} R_j^c$.
\end{proof}

\section{An Optimal Linear Code for the Point-to-Point Finite Field Adder Channel} \label{linear}

Now, we consider the following linear code that maps a length-$k$ (row vector) message $\boldsymbol{s}\in \mathcal{F}^k$ to a length-$n$ (row vector) codeword $\boldsymbol{x} \in \mathcal{F}^n$:
\begin{equation}\label{eq:linear-codes-def-1}
\boldsymbol{x} = ( \boldsymbol{s} \odot \mathbb{G} ) \oplus \boldsymbol{q}
= \left( \boldsymbol{s} \odot \begin{bmatrix} \boldsymbol{g}_1 \\ \boldsymbol{g}_2 \\ \vdots \\ \boldsymbol{g}_k \end{bmatrix} \right) \oplus \boldsymbol{q},
\end{equation}
where $\odot$ is the multiplication associated with $\mathcal{F}$, 
$\mathbb{G}$ is a fixed $k \times n$ matrix, with each element independently and uniformly chosen over $\mathcal{F}$, the $i$-th row in $\mathbb{G}$, $\boldsymbol{g}_i$, is a row vector of length $n$, and $\boldsymbol{q}$ is a fixed row vector of length $n$, with each element independently and uniformly chosen over $\mathcal{F}$.

We extend the results for binary linear codes \cite[p. 206--207]{gallager68} to finite field linear codes in the following two lemmas.

\begin{lemma}\label{lemma:linear-codes-1}
Consider the linear codes defined in \eqref{eq:linear-codes-def-1}. Over the ensemble of codes, the probability that a message $\boldsymbol{s}_1$ is mapped to a given codeword $\boldsymbol{x}_1$ is $p(\boldsymbol{x}_1)=|\mathcal{F}|^{-n}$.
\end{lemma}

\begin{proof}[Proof of Lemma~\ref{lemma:linear-codes-1}]
There are $|\mathcal{F}|^{n(k+1)}$ ways of selecting $\mathbb{G}$ and $\boldsymbol{q}$. As the elements are arbitrarily chosen, each unique $(\mathbb{G},\boldsymbol{q})$ has a probability of $|\mathcal{F}|^{-n(k+1)}$ of being selected. For any $\mathbb{G}$, there is only one $\boldsymbol{q}$ that results in the given $\boldsymbol{x}_1$. So, there are only $|\mathcal{F}|^{nk}$ different $(\mathbb{G},\boldsymbol{q})$ that map $\boldsymbol{s}_1$ to $\boldsymbol{x}_1$. Hence, $p(\boldsymbol{x}_1) = |\mathcal{F}|^{nk} |\mathcal{F}|^{-n(k+1)} = |\mathcal{F}|^{-n}$.
\end{proof}

\begin{lemma}\label{lemma:linear-codes-2}
Consider the linear codes defined in \eqref{eq:linear-codes-def-1}. Let $\boldsymbol{s}_1$ and $\boldsymbol{s}_2$ be two different messages. The corresponding codewords $\boldsymbol{x}_1 = (\boldsymbol{s}_1 \odot \mathbb{G}) \oplus \boldsymbol{q}$ and $\boldsymbol{x}_2 = (\boldsymbol{s}_2 \odot \mathbb{G}) \oplus \boldsymbol{q}$ are independent.
\end{lemma}

\begin{proof}[Proof of Lemma~\ref{lemma:linear-codes-2}]
To show independence, we need to find the probabilities $p(\boldsymbol{x}_1)$ and $p(\boldsymbol{x}_2|\boldsymbol{x}_1)$. Equivalently, we find the probabilities $p(\boldsymbol{x}_1 \oplus -\boldsymbol{x}_2)$ and $p(\boldsymbol{x}_1|\boldsymbol{x}_1 \oplus -\boldsymbol{x}_2)$, where $-\boldsymbol{x}_2$ is the \emph{additive inverse} of $\boldsymbol{x}_2$ in $\mathcal{F}$. Let $\boldsymbol{s}_1$ and $\boldsymbol{s}_2$ differ in the $j$-th position (they may differ, additionally, in other positions). So, $\boldsymbol{x}_1 \oplus -\boldsymbol{x}_2 = (\boldsymbol{s}_1 \oplus -\boldsymbol{s}_2) \odot \mathbb{G}$. For any $(\boldsymbol{g}_1,\dotsc,\boldsymbol{g}_{j-1},\boldsymbol{g}_{j+1},\dotsc,\boldsymbol{g}_k)$, there is only one $\boldsymbol{g}_j$ that results in the given $(\boldsymbol{x}_1 \oplus -\boldsymbol{x}_2)$. Hence, there are only $|\mathcal{F}|^{n(k-1)}$ different $\mathbb{G}$'s that give $(\boldsymbol{x}_1 \oplus  - \boldsymbol{x}_2)$. In addition, for any chosen $\mathbb{G}$, there is only one $\boldsymbol{q}$ that results in the given $\boldsymbol{x}_1$. So, there are only $|\mathcal{F}|^{n(k-1)}$ unique $(\mathbb{G},\boldsymbol{q})$'s that give the desired $\boldsymbol{x}_1$ and $\boldsymbol{x}_2$. So, the probability $p(\boldsymbol{x}_1,\boldsymbol{x}_2) = |\mathcal{F}|^{n(k-1)} |\mathcal{F}|^{-n(k+1)} = |\mathcal{F}|^{-2n} = p(\boldsymbol{x}_1)p(\boldsymbol{x}_2)$.
\end{proof}

With the above lemmas, we have the following theorem:
\begin{theorem}\label{theorem:optimal-code}
Consider the finite field adder channel
\begin{equation}
Y = X \oplus N, \label{eq:ff-ptp-channel}
\end{equation}
where $Y,X,N \in \mathcal{F}$, where $X$ is the channel input, $Y$ is the channel output, $N$ is independent and identically distributed (i.i.d.) noise for each channel use.
A transmitter sends a message $\boldsymbol{s} \in \mathcal{F}^k$ over $n$ uses of the channel~\eqref{eq:ff-ptp-channel} using the linear code in~\eqref{eq:linear-codes-def-1}. The receiver can reliably decode the message from the $n$ received signals $\boldsymbol{Y}$ if $n$ is sufficiently large and if
\begin{equation}\label{eq:linear}
(k\log_2|\mathcal{F}|)/n < \log_2 |\mathcal{F}| - H(N).
\end{equation}
\end{theorem}

\begin{proof}[Sketch of proof for Theorem~\ref{theorem:optimal-code}]
From Lemma~\ref{lemma:linear-codes-1} we know that for the code defined in \eqref{eq:linear-codes-def-1}, for any codeword, each \emph{codeletter} is uniform and i.i.d.. From Lemma~\ref{lemma:linear-codes-2}, we know that any pair of codewords are independent of each other. Using these two facts, we can repeat the analysis of the probability of error in the proof of the channel coding theorem~\cite[p. 201--204]{coverthomas06} to show that the receiver can decode the message $\boldsymbol{s}$ from the $n$ received signals $\boldsymbol{Y}$ with an arbitrarily small error probability if $n$ is sufficiently large and if $\frac{k\log_2|\mathcal{F}|}{n} < I(X;Y)$, where $X$ is uniformly distributed.
\end{proof}

\section{Functional-Decode-Forward with Rate Splitting and Joint Source-Channel Coding}\label{fdf}
In this section we derive an achievable rate region using the linear code derived in the previous section. Consider each user $i$, for $1 \leq i \leq L$, sending $T$ messages (of $nR_i$ bits each), denoted by $(W_i[1], W_i[2], \dotsc, W_i[T])$. Consider a total of $(T+1)n$ channel uses. Since we consider full data exchange, user $i$ needs to decode the messages sent by all the other users, i.e., $\big\{W_{j}[t]: \forall j \in [1,L] \setminus \{i\}, \forall t \in [1,T] \big\}$. Define each $n$ channel uses as a block. In the $t$-th block, for $1 \leq t \leq T$, each user $i$ sends $\boldsymbol{X}_i(W_i[t])$ on the uplink. In the $(t+1)$-th block, for $1 \leq t \leq T$, the relay transmits $\boldsymbol{X}_0$, a function of its received signals in the $t$-th block, on the downlink. At the end of the $(t+1)$-th block, each user $i$ then decodes the $t$-th message of all other users, i.e., $(W_1[t], \dotsc, W_{i-1}[t], W_{i+1}[t], \dotsc, W_L[t])$. So, for each pair of the $t$-th block on the uplink and the $(t+1)$-th block on the downlink, if each user can reliably decode the $t$-th message of all other users, then repeating the same coding scheme for all $1 \leq t \leq T$, all users can reliably decode the messages from all other users in all blocks. This means that the rate tuple $\left(\frac{TnR_1}{(T+1)n},\frac{TnR_2}{(T+1)n},\dotsc,\frac{TnR_L}{(T+1)n}\right)$ is achievable. For any $n$, $R_1$, $R_2$, $\dotsc$, $R_L$, we can choose a sufficiently large $T$ such that the achievable rate tuple is arbitrarily close to $(R_1,R_2,\dotsc,R_L)$. In this section, we derive constraints on $R_1$, $R_2$, $\dotsc$, $R_L$ such that the rate tuple is achievable.

Since the encoding and decoding functions for all nodes are repeated in each block, we focus on the first block on the uplink and the second block on the downlink. For simplicity, we denote $W_i[1]$ by $W_i$ in the rest of this section.

\subsection{Uplink}\label{uplink}

Recall that $R_i^c = \sum_{j=1, j \neq i}^LR_j$, $R_\text{min} = \min\limits_{1 \leq j \leq L} R_j$, and $R_\text{min}^c = \left(\sum_{j=1}^L R_j\right) - R_\text{min}$. 
For the uplink of the MWRC, we use the idea of FDF in \cite{ongjohnsonkellett10cl} and rate splitting. Let $R_i = R_\text{min} + R_i'$. So, each message $W_i$ can be split into $W_i = (A_i,B_i)$, where $A_i$ is $nR_\text{min}$ bits long and $B_i$ is $nR_i'$ bits long. Let $D$, $0 \leq D  < L$, be the number of users whose message is strictly more than $nR_\text{min}$ bits long. Let these users be $\{d_1,d_2,\dotsc,d_D\} \triangleq \mathcal{D} = \{j: R_j' > 0\}$. So, for all users $i \notin \mathcal{D}$, $W_i = A_i$ and $R_i'=0$.

The $n$ uplink channel uses are further split into $(L+D-1)$ sub-blocks. Each of the $t$-th sub-blocks for $1 \leq t \leq L-1$ consists of $nR_\text{min}/R_\text{min}^c$ channel uses\footnote[2]{Since $R_\text{min}$, $R_\text{min}^c$, and $R_{d_{t-L+1}}'$ are rational numbers, there exists a (possibly large) $n$ such that $nR_\text{min}/R_\text{min}^c$ and $nR_{d_{t-L+1}}'/R_\text{min}^c$ are integers.}. The $t$-th block for $L \leq t \leq L+D-1$ consists of $nR_{d_{t-L+1}}'/R_\text{min}^c$ channel uses\footnotemark[2]. Note that if we the sum the number of channel uses in all sub-blocks, we get $(L-1)nR_\text{min}/R_\text{min}^c + n\sum_{d \in \mathcal{D}}R_d'/R_\text{min}^c = n[\sum_{j=1}^L (R_\text{min} + R_j')- R_\text{min}]/R_\text{min}^c = n$.

In the $t$-th sub-block for $1 \leq t \leq L-1$, only two users transmit, using the linear code defined in \eqref{eq:linear-codes-def-1}:
\begin{equation}
\boldsymbol{X}_i = \begin{cases}
(\boldsymbol{s}(A_i) \odot \mathbb{G}_A) \oplus \boldsymbol{q}_{A,i}, &\text{if } i=t \text{ or } t+1\\
\boldsymbol{\mathfrak{0}}, &\text{otherwise},
\end{cases}
\end{equation}
where each $\boldsymbol{s}(A_i)$ is a row vector of length $k_A$, $\mathbb{G}_A$ is a fixed $k_A \times nR_\text{min}/R_\text{min}^c$ matrix, each $\boldsymbol{X}_i$ and $\boldsymbol{q}_{A,i}$ is a row vector of length $nR_\text{min}/R_\text{min}^c$, and $\boldsymbol{\mathfrak{0}}$ is the all-zero row vector (where ``zero'', $\mathfrak{0} \in \mathcal{F}$, is the additive identity of the field $\mathcal{F}$). If we say that a user $i$ \emph{does not transmit}, it sends $X_i =\mathfrak{0}$. 
$k_A$ is chosen such that
\begin{equation}\label{eq:mapping-1}
(k_A\log_2|\mathcal{F}|)/n \geq R_\text{min},
\end{equation}
so that we can define an injective (one-to-one) function that maps each $A_i$ (of $nR_\text{min}$ bits) to a unique $\boldsymbol{s}(A_i) \in \mathcal{F}^{k_A}$.


In the $t$-th sub-block for $L \leq t \leq L+D-1$, only one user, $d_{t-L+1} \in \mathcal{D}$, transmits using the linear code defined in \eqref{eq:linear-codes-def-1}:
\begin{align}
\boldsymbol{X}_i = \begin{cases}
(\boldsymbol{s}(B_i) \odot \mathbb{G}_{B,i}) \oplus \boldsymbol{q}_{B,i}. &\text{if } i = d_{t-L+1}\\
\boldsymbol{\mathfrak{0}}, &\text{otherwise},
\end{cases}
\end{align}
where $\boldsymbol{s}(B_{d_{t-L+1}})$ is a row vector of length $k_{B,d_{t-L+1}}$, $\mathbb{G}_{B,d_{t-L+1}}$ is a fixed $k_{B,d_{t-L+1}} \times nR_{d_{t-L+1}}'/R_\text{min}^c$ matrix, and each $\boldsymbol{X}_{d_{t-L+1}}$ and  $\boldsymbol{q}_{B,d_{t-L+1}}$ is a fixed row vector of length $nR_{d_{t-L+1}}'/R_\text{min}^c$.
Similarly, $k_{B,d_{t-L+1}}$ is chosen such that
\compress
\begin{equation}\label{eq:mapping-2}
(k_{B,d_{t-L+1}} \log_2|\mathcal{F}|)/n \geq R_{d_{t-L+1}}',
\end{equation}
so we can define an injective function that maps each $B_{d_{t-L+1}}$ (of $nR_{d_{t-L+1}}'$ bits) to a unique $\boldsymbol{s}(B_{d_{t-L+1}}) \in \mathcal{F}^{k_{B,d_{t-L+1}}}$.


Each element in $\mathbb{G}_A$, $\mathbb{G}_{B,d_{t-L+1}}$, $\boldsymbol{q}_{A,i}$, and $\boldsymbol{q}_{B,d_{t-L+1}}$  is independently and uniformly chosen over $\mathcal{F}$, and is fixed for all transmissions. 

In the $t$-th sub-block for $1 \leq t \leq L-1$, the relay receives $\boldsymbol{Y}_0 = \boldsymbol{X}' \oplus \boldsymbol{N}_0$, where
\compress
\begin{equation}
\boldsymbol{X}' = \Big([\boldsymbol{s}(A_t) \oplus \boldsymbol{s}(A_{t+1})] \odot \mathbb{G}_A \Big) \oplus (\boldsymbol{q}_{A,t} \oplus \boldsymbol{q}_{A,t+1}),
\end{equation}
which is also a linear codeword of the form \eqref{eq:linear-codes-def-1}. From Theorem~\ref{theorem:optimal-code}, if $nR_\text{min}/R_\text{min}^c$ is large enough and if
\compress
\begin{equation}
\frac{k_A \log_2|\mathcal{F}|}{nR_\text{min}/R_\text{min}^c} < \log_2|\mathcal{F}| - H(N_0),\label{eq:uplink-1}
\end{equation}
then the relay can reliably decode the ``message'' $\boldsymbol{s}(A_t) \oplus \boldsymbol{s}(A_{t+1}) \triangleq \boldsymbol{s}(A_{t,t+1})$.

In the $t$-th sub-block for $L \leq t \leq L+D-1$, since only one user transmits, we directly apply Theorem~\ref{theorem:optimal-code}. So, if 
\compress
\begin{equation}
\frac{k_{B,d_{t-L+1}} \log_2|\mathcal{F}|}{nR_{d_{t-L+1}}'/R_\text{min}^c} < \log_2|\mathcal{F}| - H(N_0), \label{eq:uplink-2}
\end{equation}
then the relay can reliably decode $\boldsymbol{s}(B_{d_{t-L+1}})$.


Define $U \triangleq \big(\boldsymbol{s}(A_{1,2}), \boldsymbol{s}(A_{2,3}), \dotsc, \boldsymbol{s}(A_{L-1,L}), \boldsymbol{s}(B_{d_1}),$ $\boldsymbol{s}(B_{d_2}),\dotsc,  \boldsymbol{s}(B_{d_D})\big)$. On the uplink, if
\compress
\begin{equation}
R_\text{min}^c < \log_2|\mathcal{F}| - H(N_0), \label{eq:uplink-3}
\end{equation}
we can always find sufficiently large $n$, $k_A$, and $\{k_{B,d}\}_{d \in \mathcal{D}}$, so that \eqref{eq:mapping-1}, \eqref{eq:uplink-1} and \eqref{eq:mapping-2}, \eqref{eq:uplink-2} can be satisfied in their respective sub-blocks. Hence, the relay can reliably decode $U$.


\subsection{Downlink}\label{downlink}

Assume that the relay has correctly decoded $U$. Using the strategy of joint source-channel decoding over broadcast channels~\cite{tuncel06}, the relay re-encodes $U$ and sends it on $n$ downlink channel uses. Each user $i$, for $i \in \mathcal{D}$, uses its \emph{side information} $\boldsymbol{s}(B_i)$ to decode $U$ (hence joint source-channel decoding). The users do not need to use their respective $A_i$ in the decoding, as each $A_i$ conveys little information about $U$. All users can reliably decode $U$ if~\cite[Theorem 6]{tuncel06}
\compress
\begin{align}
H(U|\boldsymbol{s}(B_i)) < nI(X_0;Y_i), \quad &\forall i \in \mathcal{D} \label{eq:down-mwrc-1}\\
H(U) < nI(X_0;Y_i),\quad &\forall i \notin \mathcal{D}, \label{eq:down-mwrc-2}
\end{align}
for some $p(x_0)$. Choosing the uniform distribution for $X_0$, $I(X_0;Y_i) = \log_2|\mathcal{F}| - H(N_i)$.

Since the mapping from $B_i$ (a random $nR_i'$-bit message) to $\boldsymbol{s}(B_i)$ is injective, $H(\boldsymbol{s}(B_i)) = H(B_i) = nR_i'$. Since $\boldsymbol{s}(A_{i,i+1}) \in \mathcal{F}^{k_A}$, $H(\boldsymbol{s}(A_{i,i+1})) \leq k_A\log_2|\mathcal{F}|$, with equality iff $\boldsymbol{s}(A_{i,i+1})$ is uniformly distributed in $\mathcal{F}^{k_A}$. From Sec.~\ref{uplink},  $(k_A\log_2|\mathcal{F}|)/n$ can be chosen arbitrarily close to $R_\text{min}$. This gives $\frac{1}{n}H(U) \leq \frac{1}{n} [\sum_{i=1}^{L-1}H(\boldsymbol{s}(A_{i,i+1})) + \sum_{d \in \mathcal{D}}\boldsymbol{s}(B_d)]  \leq (L-1)R_\text{min} + \sum_{d \in \mathcal{D}} R_d' = R_\text{min}^c$, and $\frac{1}{n}H(U|\boldsymbol{s}(B_i)) \leq R_\text{min}^c - R_i' = ([\sum_{j=1}^LR_j] - R_\text{min} - R_i') =R_i^c$.
Note that for all $i \notin \mathcal{D}$, $R_i'=0$ and hence $R_i^c = R_\text{min}^c$.
So, if 
\compress
\begin{equation}
R_i^c < \log_2|\mathcal{F}| - H(N_i), \text{ for all } 1 \leq i \leq L, \label{eq:capacity-mwrc-2}
\end{equation}
then \eqref{eq:down-mwrc-1} and \eqref{eq:down-mwrc-2} can both be satisfied. 
Note that on the downlink, linear codes are not required.

\subsection{The Capacity of the Binary MWRC}

If the rate constraints \eqref{eq:uplink-3} and \eqref{eq:capacity-mwrc-2} are satisfied, all users are able to decode $U$ reliably. Each user $i$ then performs:
\compress
\begin{align}
&\boldsymbol{s}(A_{i+1}) = \boldsymbol{s}(A_{i,i+1}) \oplus -\boldsymbol{s}(A_i),\nonumber\\
&\boldsymbol{s}(A_{i+2}) = \boldsymbol{s}(A_{i+1,i+2}) \oplus -\boldsymbol{s}(A_{i+1}),\;\;\dotsm, \nonumber \\
&\boldsymbol{s}(A_L) = \boldsymbol{s}(A_{L-1,L}) \oplus -\boldsymbol{s}(A_{L-1}),\nonumber\\
&\boldsymbol{s}(A_{i-1}) = \boldsymbol{s}(A_{i-1,i}) \oplus -\boldsymbol{s}(A_i), \nonumber\\
&\boldsymbol{s}(A_{i-2}) = \boldsymbol{s}(A_{i-2,i-1}) \oplus -\boldsymbol{s}(A_{i-1}),\;\; \dotsm,\nonumber\\
&\boldsymbol{s}(A_{1}) = \boldsymbol{s}(A_{1,2}) \oplus -\boldsymbol{s}(A_2), \label{eq:final-decoding}
\end{align}
and obtains $(A_1,A_2,\dotsc,A_{i-1},A_{i+1},\dotsc,A_L)$. Combining these with $(B_{d_1},B_{d_2},\dotsc, B_{d_D})$, each user $i$ can reliably recover the messages of all other users, i.e., $(W_1,W_2,\dotsc,W_{i-1},W_{i+1},\dotsc,W_L)$. 

So, all rate tuples $(R_1,R_2,\dotsc,R_L)$ satisfying \eqref{eq:uplink-3} and \eqref{eq:capacity-mwrc-2} are achievable. Since the closure of this region coincides with the capacity upper bound given in Theorem~\ref{theorem:upper-bound}, we have:
\begin{theorem}\label{theorem:mwrc-capacity}
Consider an $L$-user finite field adder MWRC over $\mathcal{F}$. The capacity is given by all rate tuples $(R_1,R_2,\dotsc,R_L)$ that satisfy \eqref{eq:mwrc-ub-2} and \eqref{eq:mwrc-ub-3} for all $1 \leq i \leq L$.
\end{theorem}

\begin{remark}
The capacity-achieving FDF does not utilize the users' received signals in their transmission. Hence, \emph{feedback} does not increase the capacity of the finite field adder MWRC.
\end{remark}

\subsection{A Note on the Common-Rate Capacity}
If we consider only the common rate, $R = R_i$, $\forall i$, we have $W_i = A_i$ and $B_i = \varnothing$, $\forall i$. In this case, rate splitting is not required on the uplink to get \eqref{eq:uplink-3}. Furthermore, on the downlink, since $U = (\boldsymbol{s}(A_{1,2}), \boldsymbol{s}(A_{2,3}), \dotsc, \boldsymbol{s}(A_{L-1,L}))$ has no correlation with any $W_i$, utilizing $W_i$ does not help the user in decoding $U$. On the downlink, the relay encodes $U$, of $n(L-1)R$ bits, and transmits it in $n$ channel uses. Treating the downlink from the relay to each user $i$ as a point-to-point channel~\cite[p. 200]{coverthomas06}, if $n(L-1)R < nI(X_0;Y_i)$, then user $i$ can reliably decode $U$ from its received signals without needing to use its own message (separate source-channel decoding). Hence, we get \eqref{eq:capacity-mwrc-2}. Of course, after decoding $U$, each user needs to use its message to obtain the other users' messages using the steps in  \eqref{eq:final-decoding}. But as far as channel decoding is concerned, the source messages need not be used.
So, if we are only interested in the common rate case, FDF without rate splitting and separate source-channel coding is optimal (capacity-achieving) for the finite field adder MWRC.

\section{Comparison of Coding Strategies}\label{compare}
Now, we compare three coding strategies for the special case when $L=2$ and $\mathcal{F} = \{0,1\}$, i.e., the binary TWRC. For binary $N_i$, we denote $\Pr\{N_i=1\} = \rho_i$ and $H(\alpha)= - \alpha\log_2\alpha - (1-\alpha)\log_2(1-\alpha)$.

\subsection{FDF with joint source-channel coding}
From Theorem~\ref{theorem:mwrc-capacity}, FDF with rate splitting and joint source-channel coding achieves the capacity given by $\{(R_1,R_2): R_1,R_2 \leq 1 - H(\rho_0), R_1 \leq 1 - H(\rho_2), R_2 \leq 1 - H(\rho_1) \}$. The capacity of the binary TWRC was reported in \cite{knopp07,namchung08}.

\subsection{FDF with separate source-channel coding}

Now, we find the achievable rate region using FDF with rate splitting but with \emph{separate} source-channel coding. The coding on the uplink is the same as that in  Sec.~\ref{uplink}. Assuming $R_2 \geq R_1$, we have $W_1 = A_1$ and $W_2 = (A_2,B_2)$. So, on the uplink, if $R_2 < 1 - H(\rho_0)$, then the relay can reliably decode $(\boldsymbol{s}(A_{1,2}),\boldsymbol{s}(B_2))$. Instead of using the joint source-channel coding for the downlink described in Sec.~\ref{downlink}, we re-cast the downlink as a \emph{broadcast channel with degraded message sets}, where the relay broadcasts a common message $\boldsymbol{s}(A_{1,2})$ to both the users, and a private message $\boldsymbol{s}(B_2)$ to user 1, and the users do not use their own messages for decoding  $\boldsymbol{s}(A_{1,2})$ and $\boldsymbol{s}(B_2)$ (hence separate source-channel decoding). 
From \cite{kornermarton77}, if $R_1 < 1 - H\big(\beta(1-\rho_2) + (1-\beta)\rho_2\big)$, $R_2' < H\big(\beta(1-\rho_1) + (1-\beta)\rho_1\big) - H(\rho_1)$, and $R_1 + R_2' < 1 - H(\rho_1)$, for some $0 \leq \beta \leq \frac{1}{2}$, then both users can reliably decode $\boldsymbol{s}(A_{1,2})$ and user 1 can reliably decode $\boldsymbol{s}(B_2)$ purely from their respective $\boldsymbol{Y}_i$. The users then follow the steps in \eqref{eq:final-decoding} to obtain the other user's message. Repeating this for the case $R_1 \geq R_2$, the achievable rate region is the convex hull of:
\begin{itemize}
\item $\mathcal{R}_1$: all rate pairs $(R_1,R_1 + R_2')$ satisfying
\vspace{-1ex}
\begin{align}
R_1 &< 1 - H\big(\beta(1-\rho_2) + (1-\beta)\rho_2\big) \label{eq:beta-r1}\\
R_2' &< H\big(\beta(1-\rho_1) + (1-\beta)\rho_1\big) - H(\rho_1)\\
R_1 + R_2' &< 1 - \max\{H(\rho_0), H(\rho_1)\},
\end{align}
for some $0 \leq \beta \leq \frac{1}{2}$, and
\item $\mathcal{R}_2$: all rate pairs $(R_2 + R_1',R_2)$ satisfying
\vspace{-1ex}
\begin{align}
R_2 &< 1 - H\big(\alpha(1-\rho_1) + (1-\alpha)\rho_1\big) \label{eq:beta-r2}\\
R_1' &< H\big(\alpha(1-\rho_2) + (1-\alpha)\rho_2\big) - H(\rho_2)\\
R_2 + R_1' &< 1 - \max\{H(\rho_0), H(\rho_2)\},
\end{align}
for some $0 \leq \alpha \leq \frac{1}{2}$. 
\end{itemize}


\subsection{Complete-Decode-Forward}
Using CDF, the relay fully decodes both $W_1$ and $W_2$ on the uplink, which is a multiple-access channel. So, if $R_1 < 1 - H(\rho_0)$, $R_2 < 1 - H(\rho_0)$, $R_1 + R_2 < 1 - H(\rho_0)$, then the relay can reliably decode $W_1$ and $W_2$~\cite{ahlswede71,liao72}. Note that the last inequality implies the first two.
Assuming that the relay has successfully decoded $W_1$ and $W_2$, it broadcasts $(W_1,W_2)$ on the downlink. Using a joint source-channel decoding, each user $i$, $i=1$, $2$, can reliably decode the other user's message from their respective received signals $\boldsymbol{Y}_i$ and their own messages $W_i$ if $R_1 < 1 - H(\rho_2)$ and $R_2 < 1 - H(\rho_1)$~\cite{kramershamai07,oechteringschnurr08}.
Combining the uplink and the downlink constraints, the achievable rate region using CDF is all $(R_1,R_2)$ satisfying:
\vspace{-1ex}
\begin{align}
&R_1 < 1 - H(\rho_2),\; R_2 \leq 1 - H(\rho_1),\label{eq:fdf-2} \\
&R_1 + R_2 < 1 - H(\rho_0). \label{eq:fdf-3}
\end{align}

\subsection{Discussion}
Using CDF, the relay needs to fully decode the users' messages on the uplink, and this restricts the sum rate to be constrained by the uplink bandwidth, c.f. \eqref{eq:fdf-3}. So, CDF is not \emph{uplink optimized}. On the other hand, using FDF with rate splitting and separate source-channel coding, the users' \emph{a priori} knowledge about their own messages is not utilized during the channel decoding on the downlink -- their own messages are used only \emph{after} channel decoding. So, FDF with separate source-channel coding is not \emph{downlink optimized}. These two coding strategies do not achieve the capacity of the finite field adder MWRC in general. FDF with rate splitting and joint source-channel coding overcomes these two shortcomings by having the relay decode only functions of the source messages on the uplink and having the users utilize their own messages in channel decoding on the downlink. This strategy indeed achieves the capacity of the finite field adder MWRC. This work suggests that for the general MWRC, functional decoding and joint source-channel coding should be utilized.




\end{document}